\documentclass[preprint]{elsarticle}
\usepackage[a4paper,scale=0.8]{geometry}
\usepackage{amsmath}
\usepackage{amsthm}
\usepackage{amssymb}
\usepackage{bm}  
\usepackage{url}
\usepackage[numbers]{natbib}
\usepackage{enumerate}
\usepackage{graphicx}
\usepackage{blkarray}
\usepackage{color}

\newcommand{\conv}{\mathop{\mathrm{conv}}}

\newcommand{\xc}{\mathop{\mathrm{xc}}} 

\newcommand{\RR}{{\mathbb{R}}}

\title{A generalization of extension complexity that captures $P$}

\author[mu,ku]{David Avis} 
\ead{avis@cs.mcgill.ca}

\author[cu]{Hans Raj Tiwary\corref{cor1}}
\ead{hansraj@kam.mff.cuni.cz}

\cortext[cor1]{Corresponding author}

\address[mu]{GERAD and School of Computer Science, McGill University,
   3480 University Street, Montreal, Quebec, Canada H3A 2A7.}

\address[ku]{Graduate School of Informatics,
   Kyoto University, Sakyo-ku, Yoshida Yoshida, Kyoto 606-8501, Japan}

\address[cu]{Department of Applied Mathematics (KAM) and Institute of Theoretical Computer Science (ITI),
 Charles University,
 Malostransk\'e n\'am. 25,
 118 00 Prague 1, Czech Republic}

\theoremstyle{plain}
\newtheorem{theorem}{Theorem}

\newtheorem{lemma}{Lemma}

\theoremstyle{definition}

\theoremstyle{remark}

\let\cite\citep

\BAnewcolumntype{s}{>{$\small$}c}

\newlength{\normalparindent}
\newlength{\normalparskip}
\begin{document}

\begin{abstract}

In this paper we propose a generalization of the extension complexity of a polyhedron $Q$.
On the one hand it is general enough so that all problems in $P$ can be formulated
as linear programs with polynomial size extension complexity.
On the other hand it still allows non-polynomial lower bounds to be proved for $NP$-hard problems
independently of whether or not $P=NP$.
The generalization, called $H$-free extension complexity, allows for a set of valid inequalities $H$
to be excluded in computing the extension complexity of $Q$.
We give results on the $H$-free extension complexity of hard matching problems (when $H$ are the 
odd set inequalities) and the traveling salesman problem (when $H$ are the subtour elimination constraints).

\end{abstract}

\begin{keyword}
 Polytopes \sep Extended Formulations \sep extension complexity \sep lower bounds \sep linear programming

\end{keyword}

\maketitle

\section{Introduction}
\label{intro}
Since linear programming is in $P$ we will not be able to solve an $NP$-hard problem 
$X$ in polynomial time (poly-time) by linear programming unless $P=NP$. On the other hand, since 
linear programming is $P$-complete, we will not be able to prove a super-polynomial lower bound on solving 
$X$ by a linear program (LP) without showing that $P \neq NP$. 
One way to make progress on this problem is to consider restricted versions of linear programming which have two properties:

\begin{description}
\item[Property (1):]
Problems in $P$ will still be solvable in the poly-time even in the restricted version of linear programming. 

\item[Property (2):]
Known $NP$-hard problems with natural LP formulations will have provable super-polynomial lower 
bounds under the restricted version of linear programming.
\end{description}

Note that results of type (1) and (2) will still be true, independently of whether or not $P=NP$. 

A candidate for such a restricted LP model is extension complexity. 
In formulating optimization problems as LPs, adding extra variables can greatly reduce the size of the 
LP \cite{ConfortiCornuejolsZambelli10}. An extension of a polytope is such a formulation that projects onto the
original LP formulation of the problem. In this model, LP formulations of some problems in $P$ 
that have exponential size can be reduced to polynomial size in higher dimensions.
For example Martin \cite{Martin91} showed that the minimum spanning tree problem has an extended formulation of
size $O(n^3)$ even though its natural formulation requires exponentially many inequalities.

Various authors have shown that extended formulations of various $NP$-hard problems have exponential lower bounds on their size \cite{Ya91,FMPTW,AT13, PV13}.
However this promising restricted model for LP unfortunately does not satisfy property (1):
Rothvo{\ss} \cite{Rothvoss13} recently proved that the matching problem has exponential extension complexity.

Here we propose a stronger version of extension complexity which satisfies property (1). 
We also exhibit some $NP$-hard problems that satisfy property (2).
In the proposed model we concentrate on the separation problem rather than the polynomial time equivalent optimization problem.

Let $Q$ be a polytope with half-space representation $F(Q)$ and let $H$ be a valid set of inequalities for $Q$.
We delete from $F(Q)$ all half-spaces that are redundant with respect to $H$ and call the resulting (possibly
empty) polyhedron $Q_H$. The {\em H-free extension complexity of Q} is defined to be the extension complexity
of $Q_H$.
If this extension complexity is polynomial and $H$ can be separated in poly-time then
we can solve LPs over $Q$ in poly-time. 
Note that this is true {\em even if $H$ itself has super-polynomial extension complexity}.
On the other hand if $Q_H$ has super-polynomial extension complexity,
then even if $H$ can be poly-time separated, any LP that requires an explicit formulation of $Q_H$
will have super-polynomial size. This allows us to strengthen existing results on the extension complexity
of $NP$-hard problems.
To illustrate this, we give results on the $H$-free extension complexity of hard matching problems (when $H$ are the 
odd set inequalities) and the traveling salesman problem (when $H$ are the subtour elimination inequalities).

\section{Background}

We begin by recalling some basic definitions related to extended formulations of polytopes.
The reader is referred to 
\cite{ConfortiCornuejolsZambelli10, FFGT12}
for more details.
An \emph{extended formulation} (EF) of a polytope $Q \subseteq \RR^d$
is a linear system
\begin{equation} \label{eq:EF}
E x + F y = g,\ y \geqslant \mathbf{0}
\end{equation}
in variables $(x,y) \in \RR^{d+r},$ where $E, F$ are real
matrices with $d, r$ columns respectively, and $g$ is a column vector,
such that $x \in Q$ if and only if there exists $y$ such
that \eqref{eq:EF} holds. The \emph{size} of an EF is defined as its
number of \emph{inequalities} in the system.

An \emph{extension} of the polytope $Q$ is another polytope
$Q' \subseteq \mathbb{R}^e$ such that $Q$ is the image of $Q'$ under a linear map.
We define the \emph{size} of an extension $Q'$ as the number of facets of $Q'$. 
Furthermore, we define the \emph{extension complexity} of $Q$, denoted by $\xc{( Q )},$ 
as the minimum size of any extension of $Q.$

For a matrix $A$, let $A_i$ denote the $i$th row of $A$ and $A^j$ to denote the $j$th column of $A$. Let $Q = \{x \in \RR^d \mid Ax \leqslant b\} = \conv(V)$ be a polytope, with $A \in \RR^{m \times d}$, $b \in \RR^m$ and $V = \{v_1,\ldots,v_n\} \subseteq \RR^d$. Then \(M \in \RR_+^{m \times n}\) defined as
\(M_{ij} := b_i - A_i v_j\) with \(i \in [m] := \{1,\ldots,m\}\) and \(j \in [n] := \{1,\ldots,n\}\) is the \emph{slack matrix} of \(Q\) w.r.t.\ $Ax \leqslant b$ and $V$.

We call the submatrix of $M$ induced by rows corresponding to facets and columns corresponding to
vertices the \emph{minimal slack matrix} of $Q$ and denote it by~$M(Q)$. Note that the slack matrix
may contain columns that correspond to feasible points that are not vertices of $Q$ and rows that correspond to valid inequalities that are not facets of $Q$, and therefore the slack matrix of a polytope is not a uniquely defined object. However every slack matrix of $Q$ must contain rows and columns corresponding to facet-defining inequalities and vertices, respectively.

As observed in \cite{FMPTW}, for proving bounds on the extension complexity of a polytope $Q$ it suffices to take any slack matrix of $Q$. Throughout the paper we refer to the minimal slack matrix of $Q$ as \emph{the} slack matrix of $Q$ and any other slack matrix as \emph{a} slack matrix of $Q.$

\section{$H$-free extensions of polytopes}

Let $X$ be some computational problem that can be solved by an LP over a polytope $Q$.
For the applications considered in this paper, it is convenient
to consider the case where $Q$ is given by an implicit description of its vertices.
So for the matching problem, $Q$ is the convex hull of all 0/1 matching vectors, and for the TSP problem 
it is the convex hull of all 0/1 incidence vectors of Hamiltonian circuits.

For the given polytope $Q$ let $F(Q)$ be a non-redundant half-space representation.
If $Q$ has full dimension $F(Q)$ is unique and each half-space supports a facet of $Q$.
Otherwise we may assume that $F(Q)$ is defined relative to some canonical representation
of the linearity space of $Q$.
Our restricted LP model will allow a restricted separation oracle for $Q$.

Let $H=H(Q)$ be a possibly super-polynomial size set of valid inequalities for $Q$ 
equipped with an $H$-separation oracle. 
We delete from $F(Q)$ all half-spaces that are redundant with respect to $H$ and call the resulting (possibly
empty) polyhedron $Q_H$. 

We can solve the separation problem for $Q$ for a point $x$ by first solving it for $H$
and then, if necessary, for $Q_H$. Suppose $x$ is not in $Q$.  
If $x$ is not in $H$ we get a violated inequality by the oracle. 
Otherwise $x$ must violate a facet of $Q_H$. 
We will allow separation for $Q_H$ to be performed using any extension $Q_H'$ of $Q_H$ by explicitly 
checking the facets of $Q_H'$ for the lifting of $x$. 
We call $Q'_H$ an {\em H-free EF for Q}. 
Using this separation algorithm and the ellipsoid method we have a way to solve LPs over $Q$. 
We call such a restricted method of solving LPs an {\em H-free LP for Q}.
The {\em H-free extension complexity of Q} is defined to be  $\xc(Q_H)$.
 
We say that an {\em H-free EF for Q has polynomial size} if:
\begin{description}
\item
(a) The $H$-separation oracle runs in poly-time and
\item
(b) $\xc(Q_H)$ is polynomial in the input size of $X$.
\end{description}
In this case we also have an $H$-free LP for $Q$ that runs in polynomial time.
On the other hand, if for given $H$,  $\xc(Q_H)$ is super-polynomial in the size of $X$ then we say that all $H$-free
LPs for $X$ run in super-polynomial time. Note that this statement is independent of whether or not
$P=NP$.
When $H$ is empty all of the above definitions reduce to standard definitions for EFs and extension complexity.

We illustrate these concepts with a few examples.
For the matching problem if $H$ is the set of odd-set inequalities then $Q_H$ is empty.
In this case we have an $H$-free EF for matching of poly-size even though matching has exponential extension complexity.

This example generalizes to show that every problem $X$ in $P$ has a poly-size 
$H$-free EF for some $H$. Indeed, since LP is $P$-complete, $X$ can be solved by 
optimizing over a polytope $Q$. Let $H$ be the entire facet list $F(Q)$ so that $Q_H$ is again empty. 
Optimization over $Q$ can be performed in poly-time so, by the equivalence of 
optimization and separation, separation over $H$ can be performed in poly-time also. 
Therefore (a) and (b) are satisfied as required.

For the TSP, let $H$ be the sub-tour constraints. In this case $Q_H$ is non-empty and in fact 
we will show in the next section that it has exponential extension complexity. 
Therefore $H$-free LPs for the TSP require exponential time, 
extending the existing extension complexity result for this problem. 

We remark that $H$ is an essential parameter here. Matching, for example, has 
poly-size $H$-free extension complexity when $H$ are the odd set inequalities, but not when $H$ is empty. 
Nevertheless, any problem with poly-size $H$-free extension complexity for some $H$ can of course be solved in poly-time.
For a given hard problem, one gets stronger hardness results by letting $H$ be larger 
and larger sets of poly-size separable inequalities, as long as one can still prove that 
$Q_H$ has super-polynomial extension complexity.
We give some examples to illustrate this in subsequent sections of the paper.

Before we proceed further, we note an intersection lemma that will be useful for proving lower bounds 
for $H$-free extensions of polytopes. This lemma is a polar formulation of the following result of Balas (\cite{balas85}):

\begin{lemma} \label{lem:union_lemma}
Let $P_1$ and $P_2$ be two polytopes with extension complexity $r_1$ and $r_2$ respectively. Then, the extension complexity of $\text{conv}(P_1\cup P_2)$ is at most $r_1+r_2+1.$
\end{lemma}

\begin{lemma} \label{lem:intersection_lemma}
Let $P_1$ and $P_2$ be two polytopes with extension complexity $r_1$ and $r_2$ respectively. Then, the extension complexity of $P_1\cap P_2$ is at most $r_1+r_2+1.$
\end{lemma}
\begin{proof}
If the intersection of $P_1$ and $P_2$ is not full dimensional, then we can intersect the two polytopes with a suitable affine subspace without increasing their extension complexity. So it suffices to prove this result for the case where $P_1\cap P_2$ is full dimensional.

Taking the polar dual of the two polytopes with respect to some point in the interior of $P_1\cap P_2,$ we see that $(P_1\cap P_2)^*=\text{conv}(P_1^*\cup P_2^*).$ Since the dual of a polytope has the same extension complexity, applying Lemma \ref{lem:union_lemma} we obtain the desired result.
\end{proof}

\section{The travelling salesman problem(TSP)}

An undirected TSP instance $X$ is defined by a set of integer weights $w_{ij}, 1 \le i < j \le n$,
for each edge of the complete graph $K_n$.
A tour is a Hamiltonian cycle in $K_n$ defined by a permutation of its vertices. It is
required to compute a tour of minimum weight. We define the polytope $Q$ to be the convex hull of the
0/1 incidence vectors $x=(x_{ij}: 1 \le i < j \le n\}$ of the tours. 
It is known that $\xc(Q) = 2^{\Omega(n)}$ \cite{Rothvoss13}.

We define $H$ to be the set of $subtour~elimination$ constraints:
\begin{eqnarray}
\label{subtour}
\sum_{i,j \in S, i \neq j} x_{ij} &\leq& |S|-1,~~~~S \subseteq \{1,2,...,n-1\}, ~~ |S| \ge 2.\\
x_{ij} &\ge& 0,~~~ 1 \le i < j \le n\
\end{eqnarray}

It is well known that the subtour elimination constraints can be poly-time separated by using network flows.
These constraints by themselves define the convex hull of all forests in $K_{n-1}$
and Martin \cite{Martin91} has given an EF for them that has size $O(n^3)$.

Therefore,  $xc(Q_H)=2^{\Omega(n)}$, otherwise together with Martin's result and the intersection lemma (Lemma \ref{lem:intersection_lemma}), it would imply an upper bound of $2^{o(n)}$ for the travelling salesman polytope.
It follows that every $H$-free LP for the TSP runs in exponential time, where $H$ are the subtour inequalities.

\section{Matching problems}
A matching in a graph $G=(V,E)$ is a set of edges that do not share any common vertices.
Let $n=|V|$.
A matching is called perfect if it contains exactly $n/2$ edges. For a given instance $G$ we define
the matching polytope $Q$ to be the convex hull of the 0/1 incidence vectors $x=(x_e : e\in E )$
of matchings.

For any $S \subseteq V$ and $e \in E$, we write that $e \in S$ whenever both endpoints of $e$ are in $S$.
Edmonds \cite{Edmonds1965b} proved that $Q$ has the following halfspace representation:
\begin{eqnarray}
\label{odd}
\sum_{e \in S} x_e &\leq& (|S|-1)/2,~~~~S \subseteq V, ~~ |S|~is~odd \\
0 &\le& x_e  \le ~ 1,~~~~~~ e \in E. 
\end{eqnarray}
Rothvo{\ss} \cite{Rothvoss13} recently proved that $xc(Q)=2^{\Omega(n)}$ and that a similar result holds
for the convex hull of all perfect matchings.  
Let $H$ be this half-space representation of $Q$. Since optimization over $Q$ can be performed in
poly-time by Edmonds algorithm there is a poly-time separation algorithm for $H$.
It follows that the matching problem has a poly-size $H$-free EF.

In the next three subsections we give $NP$-hard generalizations of the matching problem which have super-polynomial lower bounds
on their $H$-free extension complexity,  
where $H$ are the odd set inequalities (\ref{odd}).
The method used is similar to that described in detail in \cite{AT13}.

\subsection{Induced matchings}

A matching in a graph $G=(V,E)$ is called induced if there is no edge in $G$ between any pair of matching edges.
Stockmeyer and Vazirani \cite{Stockmeyer} and Cameron \cite{Cameron89} 
proved that the problem of finding a maximum cardinality
induced matching is $NP$-hard. Let $Q$ be the convex hull of the incidence vectors of all induced matchings in $G$.
Let $H$ be the odd  set inequalities (\ref{odd}). Clearly $H$ are valid for $Q$, and as remarked above, they admit
a poly-time separation oracle. We will prove that $\xc_H(Q)$ is super-polynomial.
Our proof makes use of the reduction in \cite{Cameron89}.

\begin{theorem}\label{thm:xc_induced_matchings}
For every $n$ there exists a bipartite graph $G$ with $O(n)$ edges and vertices such that the induced matching polytope of $G$ has extension complexity $2^{\Omega(\sqrt[4]{n})}.$
\end{theorem}
\begin{proof}
For every graph $G=(V,E)$ one can construct in polynomial time another graph 
$G'=(V',E')$ with $|V'|=2|V|$ and $|E'|=|V|+28|E|$ such that the stable set polytope of $G$ is the 
projection of a face of the induced matching polytope of $G'$ \cite{Cameron89}.
Furthermore, $G'$ is bipartite. Since, for every $n$ there exist graphs with $O(n)$ edges and 
vertices such that the stable set polytope of the graph has extension complexity 
$2^{\Omega(\sqrt[4]{n})}$ \cite{AT13}, the result follows.
\end{proof}

Since the above theorem applies to bipartite graphs $G$, each of the odd set inequalities (\ref{odd}) is 
redundant for the induced matching polytope of $G$. Therefore the $H$-free extension complexity of the induced matching polytope is super polynomial in the worst case.

Although this example offers an example $H$-free extension complexity, it suffers from one obvious weakness.
For every graph, all of the inequalities in $H$ are redundant with respect to $Q$ even
for non-bipartite graphs! 
A graph is called {\em hypomatchable} if the deletion of any vertex yields a graph with a perfect matching.
Pulleyblank proved in 1973 (see \cite{matching}) that facet-inducing inequalities in (\ref{odd})
correspond to subsets $S$ that span 2-connected hypomatchable subgraphs of $G$.
Let $x$ be the incidence vector for any matching $M$ 
in $G$ that satisfies such an inequality as an equation. Since $S$
spans a 2-connected subgraph, $M$ cannot be an induced matching.

In order to avoid such trivial cases it is desirable that most, if not all, inequalities of $H$ define facets for
at least one polytope $Q$ that corresponds to some instance of the given problem.

\subsection{Maximal matchings}
A matching in a graph $G=(V,E)$ is called maximal if its edge set is not included in a larger matching. Rather naturally, we will call the 
convex hull of the characteristic vectors of all maximal matchings of $G$ 
the maximal matching polytope of $G$ 
and denote it by $MM(G)$.  It is known that finding the minimum maximal matching is $NP$-hard \cite{YG}. 
Now we show that for every $n$ there exists a graph with $n$ vertices such that $MM(G)$ has super polynomial $H$-free extension complexity where $H$ denotes the set of odd cut inequalities.

For every 3-CNF formula $\phi$ we call the convex hull of all satisfying assignments the satisfiability polytope of $\phi$.

\begin{theorem}\label{thm:xc_restricted_sat}
For every $n$ there exists a 3-CNF formula in $O(n)$ variables such that the satisfiability polytope has super polynomial extension complexity. Furthermore, in the formula every variable appears at most twice non-negated and at most once negated.
\end{theorem}
\begin{proof}
The statement is known to be true without the restriction on the number of occurrences of the literals (\cite{FMPTW,AT13}). To impose the restriction that every variable appear at most twice non-negated and at most once negated, once can perform the following simple operations.

For any given 3-CNF formula $\phi$ construct another formula $\psi$ as follows. For each variable $x_i$, replace the occurrence of $x_i$ in clause a $C_j$ by the variable $x_i^j$ and the occurrence of $\overline{x}_i$ in clause a $C_j$ by the variable $y_i^j.$ Suppose $x_i^1,\ldots,x_i^k,y_i^1,\ldots,y_i^l$ are 
the variables replacing $x_i.$ 
We add extra clauses corresponding to the conditions $x_i^j\implies x_i^{j+1}$, that is, $\overline{x}_i^j\vee x_i^{j+1}$ for $i=1,\ldots,k-1$.
We also add more clauses corresponding to the conditions $\overline{y}_i^j\implies \overline{y}_i^{j+1}$, that is, $y_i^j\vee \overline{y}_i^{j+1}$ for $i=1,\ldots,l-1.$ 
We add two additional clauses:
$\overline{x}_i^k\vee \overline{y}_i^{1}$ 
ensures that $x_i^k\implies \overline{y}_i^{1}$ and
$y_i^l\vee x_i^{1}.$ 
ensures that $\overline{y}_i^l\implies x_i^{1}$. 
Finally since the new clauses contain two literals they are converted to clauses with three literals each
in the usual way: duplicate each clause, add a new variable to one clause and its complement
to the other. This gives a 3-CNF formula $\psi$ with 
the required properties and where the number of variables and clauses is polynomial in the size of $\phi$.

It is easy to see that the satisfiability polytope for $\phi$ is obtained by projecting the satisfiability polytope of $\psi$ along one of the variables $x_i^k$ for each $i.$ Therefore the extension complexity of the satisfiability polytope of $\psi$ is at least as high as that of $\phi$ and we have a family of 3-CNF formula with the desired restricted occurrences that have super polynomial extension complexity in the worst case.
\end{proof}

\begin{theorem}
For every $n$ there exists a bipartite graph $G=(V_1\cup V_2,E)$ such that $MM(G)$ has extension complexity super polynomial in $n$.
\end{theorem}
\begin{proof}
For every restricted 3-SAT formula $\phi$ one can construct, in polynomial time, a bipartite graph $G$ such that the maximal matching polytope $MM(G)$ is an extended formulation of the satisfiability polytope of $\phi.$ Therefore, the extension complexity of $MM(G)$ is super polynomial for the formulae used in Theorem \ref{thm:xc_restricted_sat}.
\end{proof}

Again, since the graphs $G$ in the above theorem are bipartite, each of the odd set inequalities (\ref{odd}) 
is redundant for the maximal matching polytope of $G$. Therefore the $H$-free extension complexity of the induced matching polytope is super polynomial in the worst case.

This example differs from the example in the previous subsection in that (\ref{odd})
are facet defining for maximum matching polytopes 
of non-bipartite graphs. Too see this, fix a graph $G$ and odd-set $S$ of its vertices.
Pulleyblanks's characterisation \cite{matching} states that (\ref{odd})
is facet defining for the matching polytope of $G$ whenever $S$ spans a 2-connected hypomatchable subgraph.
The only matchings in $G$ that lie on this facet 
have precisely $(|S|-1)/2$ edges from the set $S$ and
are therefore maximal on $S$. Each of these matchings can be extended to a maximal matching in $G$
which appears as a vertex of $MM(G)$. Therefore, provided these extensions do not lie in
a lower dimensional subspace and $MM(G)$ is full dimensional, (\ref{odd})
is also facet inducing for  $MM(G)$ for the given set $S$.
For example, the odd cycles $C_{2k+1},~k \ge 3$ with the addition of a chord cutting off a triangle
are a family of such graphs.

\subsection{Edge disjoint matching and perfect matching }
Given a bipartite graph $G(V_1\cup V_2,E)$ and a natural number $k$, it is $NP$-hard to decide whether $G$ contains a perfect matching $M$ and a matching $M'$ of size $k$ such that $M$ and $M'$ do not share an edge \cite{domotor}.

For a given graph $G$ with $n$ vertices and $m$ edges consider a polytope in the variables $x_1,\ldots,x_m,y_1,\ldots,y_m.$  For a subset of edges encoding a perfect matching $M$ and a matching $M'$ of size $k$ the we construct a vector with $$x_i=\begin{cases} 1, &\text{if } e_i\in M\\ 0, &\text{if } e_i\notin M\end{cases},~~~~\\y_i=\begin{cases} 1, &\text{if } e_i\in M'\\ 0, &\text{if } e_i\notin M'\end{cases}$$

Let us denote the convex hull of all the vectors encoding an edge disjoint perfect matching and a matching of size at least $k$ as $MPM(G,k).$ We would like to remark that one can also define a ``natural'' polytope here without using separate variables for a matching and a perfect matching and instead using the charectestic vectors of all subsets of edges that are an edge-disjoint union of a matching and a perfect matching. However, the formulation that we consider allows different cost functions to be applied to the matching and the perfect matching.

Now we show that for every $n$ there exists a bipartite graph $G$ with $n$ vertices and a constant $0<c<\frac{1}{2}$ such that $MPM(G,cn)$ has extension complexity super polynomial in $n$. We will use the same reduction as in \cite{domotor}, which is a reduction from MAX-2-SAT. So we first prove a super polynomial lower bound for the satisfiability polytope of 2-SAT formulas.

\begin{theorem}\label{thm:xc_2-sat}
For every $n$ there exists a 2-SAT formula $\phi$ in $n$ variables such that the satisfiability polytope of $\phi$ has extension complexity at least $2^{\Omega(\sqrt[4]{n})}.$
\end{theorem}
\begin{proof}
It was shown in \cite{AT13} that for every $n$ there exists a graph $G$ with $O(n)$ edges and vertices such that the stable set polytope of $G$ has extension complexity $2^{\Omega(\sqrt[4]{n})}.$ Since the stable sets of a graph can be encodes as a 2-SAT formula as: $\displaystyle\bigwedge_{(i,j)\in E}(\overline{x}_i\vee \overline{x}_j),$ we obtain a family of 2-SAT formulas whose satisfiability polytope for extension complexity at least $2^{\Omega(\sqrt[4]{n})}.$
\end{proof}

Note that the 2-SAT instances required in the above theorem are always satisfiable.

\begin{theorem}\label{thm:xc_mpm}
For every $n$ there exists a bipartite graph $G$ on $n$ vertices and a constant $0<c<\frac{1}{2}$ such that $MPM(G,cn)$ has super polynomial extension complexity.
\end{theorem}
\begin{proof}
The construction in \cite{domotor} implicitly provides an algorithm that given any 2-SAT formula $\phi$ with $n$ variables and $m$ clauses constructs a bipartite graph $G$ with $4mn+4m$ vertices and maximum degree 3 such that for $k=2mn+s,$ there is an assignment of variables that satisfy at least $s$ clauses of $\phi$ if and only if $G$ has and edge disjoint perfect matching and a matching of size $k$. Further the satisfiability polytope of $\phi$ is the projection of $MPM(G,k)$ and so we obtain a family of bipartite graphs with super polynomial extension complexity.
\end{proof}

Note that for every pair of odd subsets $S_1,S_2$ of $G$ two odd set inequalities can be written: one corresponding to the odd set inequalities for perfect matching polytope on variables $x_i,$ and the other corresponding to the odd set inequalities for matching polytope on variables $y_i.$ 
For a subset of vertices $S$, let $\delta(S)$ denote the subset of edges with exactly one endpoint in $S$.
The two sets of inequalities are:
\begin{eqnarray}
\label{odd_pm}
\sum_{e \in \delta(S_1)} x_e &\geqslant& 1,~~~~~~~~~~~~~S_1 \subseteq V, ~~ |S_1|~is~odd \\
\label{odd_m}
\sum_{e \in S_2} y_e &\leqslant& \frac{|S_2|-1}{2},~~~~S_2 \subseteq V, ~~ |S_2|~is~odd 
\end{eqnarray}

Again the graphs $G$ in the above theorem are bipartite so each of the odd set inequalities 
(\ref{odd_pm},\ref{odd_m}) is redundant for $MPM(G)$. Therefore taking $H$ to be the set of these inequalities we have that the $H$-free extension complexity of the these polytopes is super polynomial in the worst case. 

\section{Concluding remarks}
  \label{sect:conclusion}
We have proposed a generalization of extended formulations and extension complexity which allows
partial use of an oracle to separate valid inequalities from a specified set $H$. Our restricted
LP model allows use of the oracle and an explicit half-space representation of the remaining
non-redundant inequalities $Q_H$. In this restricted LP model, all problems in P are solvable in poly-time
even if $H$ itself has super-polynomial extension complexity. On 
the other hand, 
if $xc(Q_H)$ is super-polynomial then so is the running time of any LP in our restricted model.

This model allows for progressively stronger lower bounds as more valid inequalities are included in the
set $H$. For example, for the TSP, it would be of interest to include 
poly-time separable comb inequalities along with the subtour 
elimination inequalities in $H$ and see if one can still prove a super-polynomial bound on $xc(Q_H)$.

\section*{Acknowledgments}

Research of the first author is supported by 
a Grant-in-Aid for Scientific Research on Innovative Areas -- 
Exploring the Limits of Computation, MEXT, Japan.
Research of the second author is partially supported by the Center of Excellence -- Institute for Theoretical Computer  Science, Prague (project P202/12/G061 of GA~\v{C}R).

\bibliographystyle{plain}
\bibliography{hfree}

\end{document}